\documentclass[12pt,draftcls,peerreview,onecolumn]{IEEEtran}
\usepackage{psfrag}
\usepackage{amssymb}
\usepackage{amsmath}
\usepackage{pifont}
\usepackage{cite}
\usepackage{graphics}
\usepackage{graphicx}
\usepackage{epsfig}
\usepackage{subfigure}
\usepackage{amscd}
\usepackage{hyperref}
\usepackage{bm}

\usepackage{tikz}

\begin{document}

\title{A Simple Proof of the Mutual Incoherence Condition for Orthogonal Matching Pursuit}

\author{\IEEEauthorblockN{Jian Wang and Byonghyo Shim}\\
\IEEEauthorblockA{Information System Laboratory \\School of Information and Communication\\
Korea University,
Seoul, Korea 136-713\\
Email: jwang@ipl.korea.ac.kr, bshim@korea.ac.kr \\
Phone: 82-2-3290-4842 \\}
\thanks{This work was supported by the National Research Foundation of Korea (NRF) grant funded by the Korea government (MEST) (No. 2010-0012525) and the research grant from the second BK21 project.}
} \maketitle

\begin{abstract}
This paper provides a simple proof of the mutual incoherence condition ($\mu < \frac{1} {{2K - 1}}$)
under which $K$-sparse signal can be accurately reconstructed from a small number of linear measurements using the orthogonal matching pursuit (OMP) algorithm.
Our proof, based on mathematical induction, is built on an observation that the general step of the OMP process is in essence same as the initial step since the residual is considered as a new measurement preserving the sparsity level of an input vector.
\end{abstract}


\IEEEpeerreviewmaketitle
\setcounter{page}{1}

%
\newtheorem{thm}{Theorem}
\newtheorem{cor}[thm]{Corollary}
\newtheorem{lem}[thm]{Lemma}
\newtheorem{rem}{Remark}
\newtheorem{defi}{Definition}
\newtheorem{prop}{Proposition}

\section{Introduction}

As a sampling paradigm guaranteeing the reconstruction of sparse
signal with sampling rate significantly lower than the Nyquist rate,
compressive sensing (CS) has received considerable attention in
recent years
\cite{tropp2007signal},\cite{davenport2010analysis},\cite{candes2008restricted}.
The main goal of the CS is to accurately reconstruct a high dimensional sparse vector
using a small number linear measurements. Specifically, for a given matrix ${\mathbf{\Phi} \in {\mathbb{R}^{m \times n}} }$ ($n > m$), the CS recovery algorithm generates an estimate of $K$-sparse
vector ${\mathbf{x}} \in {\mathbb{R}^n}$ from a set of linear
measurements
\begin{eqnarray}
{\mathbf{y = \Phi x}}.
\end{eqnarray}
Although this task seems to be a severely ill-posed inverse problem, due to the prior knowledge of
sparsity information, $\mathbf{x}$ can be perfectly reconstructed via
properly designed recovery algorithm.
Among many greedy search algorithms developed for this purpose, OMP algorithm has received special attention due to its simplicity and competitive reconstruction performance \cite{tropp2004greed}.

Theoretical analysis of OMP to date has concentrated primarily on
two fronts. The first approach is based on the restricted isometric property (RIP).
A sensing matrix $\mathbf{\Phi}$ satisfies the RIP of order $K$ if there exists a constant $\delta$
such that \cite{dai2009subspace}
\begin{eqnarray} \label{eq:RIP}
\left( {1 - {\delta}} \right)\left\| {\mathbf{x}} \right\|_2^2 \leq
\left\| {{\mathbf{\Phi x}}} \right\|_2^2 \leq \left( {1 + {\delta}}
\right)\left\| {\mathbf{x}} \right\|_2^2
\end{eqnarray}
for any $K$-sparse vector $\mathbf{x}$ ($\left\|\mathbf{x}\right\|_0
\leq K$).
In particular, the minimum of all constants $\delta$
satisfying (\ref{eq:RIP}) is called the isometry constant $\delta
_K$. Wakin and Davenport have shown that the OMP can
reconstruct all $K$-sparse signals if ${\delta _{K + 1}} < \frac{1}
{{3\sqrt K }}$ \cite{davenport2010analysis}. This result has been recently improved by Wang and Shim to ${\delta _{K + 1}} < \frac{1}{{\sqrt K } + 1}$ \cite{wang_asilomar}.
The second approach is based on the coherence parameter. The coherence parameter $\mu$ of the sensing matrix $ \mathbf{\Phi} $ is defined as
$$ \mu = \mathop {\max }\limits_{i \neq j}  | \langle\varphi_i , \varphi_j \rangle| $$
where $ \varphi_i $ and $ \varphi_j $ are two column vectors of $ \mathbf{\Phi} $.
%
When the columns of $\mathbf{\Phi}$ have unit norm and satisfy the mutual
incoherence condition given by $\mu < \frac{1}{2K-1}$, the OMP will
recover $K$-sparse signal x from the measurements $\mathbf{y} =
\mathbf{\Phi} \mathbf{x}$ \cite{tropp2004greed}.
It is well known that this result is also applied to $\ell_1$-minimization approach \cite{huo2001}.

In this work, we provide a simple proof of the mutual incoherence condition for the OMP using mathematical induction.
Our proof is built on an observation that the general step of the OMP process is in essence same as the initial step since the residual is considered as a new measurement preserving the sparsity level of an input vector.
The mutual incoherence condition for the OMP is formally described in the following theorem.

\begin{table}
\begin{center}
\caption{OMP Algorithm} \label{OMP} \small
\begin{tabular}{ll} \hline \hline
Input:  &$\mathbf{y}$, $\mathbf{\Phi}$, $K$.\\
Initialize: &$k = 0$, $\mathbf{r}^{0} = \mathbf{y}$, $T^{0} = \emptyset$.\\
While &$ k < K$ \\
&$k = k + 1$.\\
&(Identify)\hspace{5.3mm}$t^{k} = \arg\max_{j}|\langle\mathbf{r}^{k-1},\mathbf{\varphi}_{j}\rangle|$.\\
&(Augment)\hspace{3.1mm}$T^{k} = T^{k-1}\cup \{t^{k}\}$.\\
&(Estimate)\hspace{4mm}${{\mathbf{\hat x}}_{{T^k}}} = \arg  \mathop
{\min}\limits_{\mathbf{x}} {\left\|\mathbf{y}-\mathbf{\Phi}_{{T^k}}
\mathbf{x}\right\|}_{2}$. \\
&(Update)\hspace{6.4mm}${{\mathbf{r}}^k} = {\mathbf{y}} -
{{\mathbf{\Phi }}_{{T^k}}}{{\mathbf{\hat
x}}_{{T^k}}}$.\\

End\\
Output:  &$\hat{\mathbf{ x}} = \arg  \mathop
{\min}\limits_{\mathbf{x}:{\text{supp}}\left(
     {\mathbf{x}} \right) = T^{{K}}}
     {\left\| {{\mathbf{y - \Phi x}}}
     \right\|_2}$.\\
\hline
\end{tabular}
\end{center}
\end{table}

%
%
\vspace{0.8cm}
\begin{thm}[{\em Mutual incoherence condition for OMP}]\label{general}
For any $K$-sparse vector $\mathbf{x}$, the OMP algorithm perfectly
recovers $\mathbf{x}$ from the measurements $\mathbf{y} = \mathbf{\Phi x}$ if the coherence parameter $\mu$ satisfies
\begin{eqnarray}
\mu  < \frac{1} {{2K - 1}}.
\end{eqnarray}
\end{thm}
\vspace{0.7cm}

\section{Simple Proof of Theorem \ref{general}}
Before presenting the proof of Theorem \ref{general}, we provide lemmas useful in our analysis.

\vspace{0.3cm}
\begin{lem}[Norm inequality \cite{golub1996matrix}] \label{lem:norm}
For $\mathbf{A} $ and $\mathbf{B} $ in $\mathbb{R}^{m\times n}$, $\alpha \in \mathbb{R}$, and $\mathbf{x} \in \mathbb{R}^n$, following inequalities are satisfied:
 $$\| \mathbf{A} \|_2 \leq \sqrt{mn}  \| \mathbf{A} \|_{\max},$$
 $$\| \mathbf{A} + \mathbf{B} \|_2 \leq \| \mathbf{A} \|_2  +  \| \mathbf{B} \|_2,$$
 $$\| \mathbf{A}   \mathbf{B} \|_2 \leq \| \mathbf{A} \|_2    \| \mathbf{B} \|_2,$$
 $$\|  \alpha  \mathbf{A} \|_2 = |\alpha|   \| \mathbf{A} \|_2,$$
 $$\|  \mathbf{A x} \|_2 \leq \| \mathbf{A} \|_2 \| \mathbf{x} \|_2,$$
where $\| \mathbf{A} \|_2$ is the spectral norm of $\mathbf{A}$ and $\| \mathbf{A} \|_{\max}$ is the maximum absolute value of elements of $\mathbf{A}$ (i.e.,  $\| \mathbf{A} \|_{\max} = \max_{i,j} | a_{i,j} |$).
\end{lem}

\vspace{0.3cm}
\begin{lem} [A direct consequence of RIP \cite{dai2009subspace}]
\label{lem:rips} 
Let $I \subset \left\{ {1,2, \cdots,n} \right\}$ and $\mathbf{\Phi }_I$ be the restriction of the columns of $\mathbf{\Phi}$ to a support set $I$.
If $\delta_{\left| I \right|} < 1$,
then for any ${\mathbf{u}} \in {\mathbb{R}^{\left| I \right|}}$,
\begin{eqnarray}
\left( {1 - {\delta _{\left| I \right|}}} \right){\left\|
{\mathbf{u}} \right\|_2}  \leq {\left\| {{\mathbf{\Phi
}}_I'{{\mathbf{\Phi }}_I}{\mathbf{u}}} \right\|_2}
 \leq \left( {1 + {\delta _{\left| I
\right|}}} \right){\left\| {\mathbf{u}} \right\|_2}. \nonumber
\end{eqnarray}
\end{lem}

\vspace{0.3cm}
\begin{lem}\label{lem:rip_mip}
The isometric constant ${\delta _K}$ for the sensing matrix $\mathbf{\Phi}$ satisfies
$${\delta _K} \leq \left( {K - 1} \right)\mu. $$
\end{lem}
\begin{proof}

Using Lemma \ref{lem:norm}, we have
\begin{eqnarray}
  \left\| {{{\mathbf{\Phi }}_T}{{\mathbf{x}}_T}} \right\|_2^2 &=& {\left\| {{\mathbf{x}}_T'{\mathbf{\Phi }}_T'{{\mathbf{\Phi }}_T}{{\mathbf{x}}_T}} \right\|_2} \nonumber   \\
&\leq&  {\left\| {{\mathbf{x}}_T'} \right\|_2}{\left\| {{\mathbf{\Phi }}_T'{{\mathbf{\Phi }}_T}} \right\|_2}{\left\| {{{\mathbf{x}}_T}} \right\|_2} \nonumber   \\
\label{eq:lemma2} &=& {\left\| {{\mathbf{\Phi }}_T'{{\mathbf{\Phi
}}_T}} \right\|_2}\left\| {{{\mathbf{x}}_T}} \right\|_2^2
\end{eqnarray}
where $T$ is the support of $\mathbf{x}$ (a set with the locations
of the non-zero elements of $\mathbf{x}$), ${{\mathbf{x}}_T}$ is a
vector composed of the elements of $\mathbf{x}$ indexed by $T$.
Noting that $(i,j)$-th element of ${{\mathbf{\Phi
}}_T'{{\mathbf{\Phi }}_T}}$ is $\langle\varphi_i , \varphi_j
\rangle$ and $\varphi_i$ is the unit norm vector, it is clear that
\begin{eqnarray}
{\mathbf{\Phi }}_T'{{\mathbf{\Phi }}_T} = \left( \begin{matrix}
1 & \langle\varphi_1 , \varphi_2 \rangle  & \cdots & \langle\varphi_1 , \varphi_n \rangle  \\
\langle\varphi_2 , \varphi_1 \rangle  & 1 & \cdots & \langle\varphi_2 , \varphi_n \rangle  \\
& & \cdots & \\
\langle\varphi_n , \varphi_1 \rangle & \langle\varphi_n , \varphi_2 \rangle & \cdots & 1
\end{matrix} \right) .
\end{eqnarray}
%
%
Now, let
\begin{eqnarray}
{{\mathbf{\Phi }}_T'{{\mathbf{\Phi }}_T}} = { (1 - \mu){\mathbf{I}}
+ {\mathbf{A}}}
\end{eqnarray}
then
\begin{eqnarray}
\mathbf{A} = \left( \begin{matrix} \mu & \langle\varphi_1 , \varphi_2 \rangle  & \cdots & \langle\varphi_1 , \varphi_K \rangle  \\
\langle\varphi_2 , \varphi_1 \rangle  & \mu & \cdots & \langle\varphi_2 , \varphi_K \rangle  \\
& & \cdots & \\
\langle\varphi_K , \varphi_1 \rangle & \langle\varphi_K , \varphi_2 \rangle & \cdots & \mu
\end{matrix} \right)
\end{eqnarray}
and thus $ {\left\| {\mathbf{A}}
\right\|_{\max }}  = \mu$. Hence,
\begin{eqnarray}
\label{eq:xx1} {\left\| {{\mathbf{\Phi }}_T'{{\mathbf{\Phi }}_T}}
\right\|_2} &=&
{\left\|{ (1 - \mu){\mathbf{I}} + {\mathbf{A}}} \right\|_2}     \\
\label{eq:xx2}
&\leq& {\left\| (1 - \mu){\mathbf{I}} \right\|_2} + {\left\| {\mathbf{A}} \right\|_2} \\
\label{eq:xx3}
&\leq&  1 - \mu  + \sqrt {K^2 } {\left\| {\mathbf{A}} \right\|_{\max }}     \\
\label{eq:xx4} &=& 1 + \left( {K - 1} \right)\mu
\end{eqnarray}
where (\ref{eq:xx3}) is from Lemma \ref{lem:norm}.
Using (\ref{eq:lemma2}) and (\ref{eq:xx4}), we have
\begin{eqnarray}
\left\| {{{\mathbf{\Phi }}_T}{{\mathbf{x}}_T}} \right\|_2^2 &\leq& {\left\| {{\mathbf{\Phi }}_T'{{\mathbf{\Phi }}_T}} \right\|_2}\left\| {{{\mathbf{x}}_T}} \right\|_2^2 \nonumber   \\
&\leq& \left( {1 + \left( {K - 1} \right)\mu } \right)\left\|
{{{\mathbf{x}}_T}} \right\|_2^2.  \nonumber
\end{eqnarray}
Recalling the definition of the RIP that ${\delta _K}$ is the minimum satisfying (\ref{eq:RIP}), we have $${\delta _K} \leq \left( {K - 1} \right)\mu.$$

\end{proof}

%
%
%
%
%
\vspace{0.5cm}
{\bf Proof of theorem \ref{general}}
\begin{proof}
We will prove the theorem using induction.
In the first iteration ($k = 1$) of the OMP algorithm, $t^k (= t^1)$
becomes the index of the column maximally correlated with the
measurement $\mathbf{y}$, i.e.,
\begin{eqnarray}
{t^k}  = \arg \max_i \left| {\left\langle {{\varphi
_i},{\mathbf{y}}} \right\rangle } \right| .
\end{eqnarray}
%
%

Then, we have
\begin{eqnarray}
\label{eq:zz1} \left| {\left\langle {{\varphi _{t^k} },{\mathbf{y}}}
\right\rangle } \right| &=& \mathop {\max }\limits_i \left|
{\left\langle {{\varphi _i},{\mathbf{y}}} \right\rangle } \right|
\\
\label{eq:zz2} &\geq& \frac{1} {{\sqrt K }}{\left\|
{{\mathbf{\Phi}}_T'{\mathbf{y}}} \right\|_2}
\\
\label{eq:zz3} &\geq& \frac{1} {{\sqrt K }}{\left\| {{\mathbf{\Phi
}}_T'{{\mathbf{\Phi }}_T}{{\mathbf{x}}_T}} \right\|_2}
\\
\label{eq:zz4}
&\geq& \frac{1} {{\sqrt K }}\left( {1 - {\delta _K}} \right){\left\|
{{{\mathbf{x}}_T}} \right\|_2}
\\
\label{eq:zz5}
&\geq& \frac{1}{{\sqrt K }}\left( {1 - \left( {K - 1} \right)\mu } \right){\left\|
{{{\mathbf{x}}_T}} \right\|_2}
\end{eqnarray}
where (\ref{eq:zz3}) is due to ${\mathbf{y}} = {{\mathbf{\Phi }}_T}{{\mathbf{x}}_T}$, (\ref{eq:zz4}) and (\ref{eq:zz5}) follow from Lemma \ref{lem:rips} and \ref{lem:rip_mip}, respectively.

Now, suppose that ${t^k} $ is not belonging to the support of
${\mathbf{x}}$ (i.e., ${t^k}  \notin T$), then
\begin{eqnarray} \label{eq:small1}
\left| \left\langle {{\varphi }_{{t^k} }},\mathbf{y} \right\rangle  \right| & = & {{\left\| \varphi _{{t^k} }'{{\mathbf{\Phi }}_{T}}{{\mathbf{x}}_{T}} \right\|}_{2}}  \label{eq:uu1} \\
 & \leq & {{\left\| \varphi _{{t^k} }'{{\mathbf{\Phi }}_{T}} \right\|}_{2}}{{\left\| {{\mathbf{x}}_{T}} \right\|}_{2}}  \label{eq:uu2} \\
 & = &\sqrt{\sum\limits_{i \in T } {{{| \left\langle\varphi _{{t^k} },{{\varphi }_{i}} \right\rangle|}^{2}}}} {\left\| {{{\mathbf{x}}_T}} \right\|_2}  \label{eq:uu3} \\
 & \leq  & \sqrt{\sum\limits_{i \in T } {{{\mu }^{2}}}} {\left\| {{{\mathbf{x}}_T}} \right\|_2}  \label{eq:uu4} \\
 & = & \sqrt{K} \mu {\left\| {{{\mathbf{x}}_T}} \right\|_2}
\end{eqnarray}
where (\ref{eq:uu4}) is from the definition of $\mu$ ($ \mu =
\mathop {\max }\limits_{i \neq j}  | \langle\varphi_i , \varphi_j
\rangle| $).
This case, however, will never occur if
\begin{eqnarray}
\frac{1} {{\sqrt K }}\left( {1 - \left( {K - 1} \right)\mu }
\right){\left\| {{{\mathbf{x}}_T}} \right\|_2} > \sqrt K \mu
{\left\| {{{\mathbf{x}}_T}} \right\|_2}
\end{eqnarray}
or
\begin{eqnarray} \label{eq:sufficient}
\mu  < \frac{1} {{2K - 1}}.
\end{eqnarray}
In summary, if $\mu  < \frac{1} {{2K - 1}}$, then ${t^k} \in T$ for
the first iteration of the OMP algorithm.

%
%
%
%
%
Now we assume that the former $k$ iterations are successful ($
T^k = \{t^1, t^2,\cdots, t^k \} \in T$) for $1 \leq k \leq K - 1$.
Then it suffices to show that $t^{k+1}$ is in $T$ but not in $T^k$ ($t^{k+1} \in T \backslash T^k$).
Recall from Table \ref{OMP} that the residual at the $k$-th iteration of the OMP is
\begin{eqnarray}
{{\mathbf{r}}^k} = {\mathbf{y}} - {{\mathbf{\Phi }}_{{T^k}}}{{\mathbf{\hat
x}}_{{T^k}}}.
\end{eqnarray}
Since ${\mathbf{y}} = {{\mathbf{\Phi }}_T}{{\mathbf{x}}_T}$ and
${{\mathbf{\Phi }}_{{T^k}}}$ is a submatrix of ${{\mathbf{\Phi
}}_T}$,
%
${{\mathbf{r}}^k} \in span\left( {{{\mathbf{\Phi }}_T}} \right)$
and thus ${{\mathbf{r}}^k}$ can be expressed as a linear combination of the $|T|$
($= K$) columns of ${{\mathbf{\Phi }}_T}$.
Accordingly, we can express ${{\mathbf{r}}^k}$ as ${{\mathbf{r}}^k} =
{\mathbf{\Phi x'}}$ where the support (set of indices for nonzero
elements) of $\mathbf{x'}$ is contained in the support of
$\mathbf{x}$. In this sense, it is natural to interpret that ${{\mathbf{r}}^k}$ is
a measurement of $K$-sparse signal $\mathbf{x'}$ using the sensing
matrix $\mathbf{\Phi}$.
Thus, if (\ref{eq:sufficient}) is satisfied, we guarantee that $t^{k+1} \in T$ at the $(k +
1)$-th iteration.
Noting that the residual ${{\mathbf{r}}^k}$ is orthogonal to the columns already selected\footnote{Since ${{\mathbf{\Phi }}_{{T^k}}}{{\mathbf{\hat x}}_{{T^k}}}$ is a projection of ${\mathbf{y}}$, the error vector ${\mathbf{y}} - {{\mathbf{\Phi }}_{{T^k}}}{{\mathbf{\hat x}}_{{T^k}}}$ (which equals ${{\mathbf{r}}^k}$) is orthogonal to the projection ${{\mathbf{\Phi }}_{{T^k}}}{{\mathbf{\hat x}}_{{T^k}}}$.}, 
index of these columns is not selected again (see the identify step in Table I) and hence $t^{k+1} \in T \backslash T^k$. This concludes the proof.
\end{proof}
%
%
%
%
%

\vspace{1.0cm}


Thus far, we have shown that the OMP algorithm is working perfectly if the sensing matrix $\mathbf{\Phi}$ satisfies the condition $\mu  < \frac{1} {{2K - 1}}$. Interestingly, this condition is not only sufficient but also necessary. We prove this claim by showing that, even with slight relaxation of this condition ($\mu  = \frac{1} {{2K - 1}}$), it is possible that the OMP algorithm cannot perfectly recover $K$-sparse signal.
Note that our construction of $\mathbf{\Phi}$ is similar to Cai, Wang, and Xu's work for proving the tightness of mutual incoherence condition for $\ell_1$-minimization \cite[Remark 3.2]{cai2010stable}.

%
%
%
%
%
%
\begin{rem}[Necessity of $\mu  <  \frac{1} {{2K - 1}}$]
Suppose $\mathbf{\Phi}$ has normalized columns (${\left\| {{\varphi _i}} \right\|_2} = \varphi _i'{\varphi _i} = 1$) and also satisfies $\mu = \frac{1} {{2K - 1}}$.
Then it is clear that ${{\mathbf{\Phi}}'}{\mathbf{\Phi }}  \in \mathbb{R}^{n \times n}$ has a unit diagonal and the absolute
value of the off-diagonal elements is upper bounded by $ \frac{1} {{2K - 1}}$.
Now consider
$${{\mathbf{\Phi }}'}\mathbf{\Phi }={\left( \begin{matrix}
   1 & -\frac{1}{2K-1} & \cdots  & -\frac{1}{2K-1}  \\
   -\frac{1}{2K-1} & 1 & \cdots  & -\frac{1}{2K-1}  \\
       &   & \cdots  &  \\
   -\frac{1}{2K-1} & -\frac{1}{2K-1}  & \cdots  &   1  \\
\end{matrix} \right)}.$$
Then ${{\mathbf{\Phi }}'}{\mathbf{\Phi }}$ is symmetric and positive
semi-definite matrix, and hence $\mathbf{\Phi }$ can be found by an eigen-decomposition of ${{\mathbf{\Phi }}'}\mathbf{\Phi
}$\cite{cai2010stable}.
Note that an $n \times n$ matrix $\mathbf{K}$ with $(\mathbf{K})_{i,i} = a$ and $(\mathbf{K})_{i,j} = b$ for $i \neq j$ is invertible if and only if $a + (n - 1)b \neq 0$.
Hence, for $a = 1$ and $b = - \frac{1}{2K-1}$, ${{\mathbf{\Phi }}'}{\mathbf{\Phi }}$ is not
invertible for the choice of $n = 2K$.
In this case, eigen-decomposition of ${{\mathbf{\Phi }}'}\mathbf{\Phi }$ becomes
\begin{eqnarray}
{{\mathbf{\Phi }}'}\mathbf{\Phi } = \mathbf{U} \mathbf{\Lambda}
 \mathbf{U}'
\end{eqnarray}
and
$$\mathbf{\Phi } =  \sqrt{ \mathbf{\Lambda } }\mathbf{U}'$$
where
\begin{equation*}
         \Lambda = \left( \begin{array}{ll}
                                                                \begin{array}{llll}
 \lambda_{1} & 0                       & \cdots  & 0                       \\
 0                       & \lambda_{2} & \cdots  & 0                       \\
                         &                         & \cdots  &                         \\
 0                       & 0                       & \cdots  & \lambda_{l}
                                                                \end{array}       &
\underbrace{
                                                                \begin{array}{llll}
 0 & \cdots  & 0  \\
 0 & \cdots  & 0  \\
   &         &    \\
 0 & \cdots  & 0
                                                                \end{array}
}_{2K - l}
                                \end{array}
           \right)
\end{equation*}
and $\{\lambda_i\}_{i = 1,2, \cdots, l}$ are $ l $ nonzero
eigenvalues of ${{\mathbf{\Phi}}'}{\mathbf{\Phi }}$.
Since the rank of $\mathbf{\Phi }$ is $l \, (< 2K)$, there exists a vector $\mathbf{z} \in \mathbb{R}^{2K}$ (which by definition is $2K$-sparse vector) in the null space of ${\mathbf{\Phi}}$ obeying ${\mathbf{\Phi z}} = {\mathbf{0}}$. One can
then divide ${\mathbf{z}}$ into two $K$-sparse vectors ${\mathbf{x}_1}$ and $-{\mathbf{x}_2}$ (i.e., $\mathbf{z} = \mathbf{x}_1 -\mathbf{x}_2$).
This gives ${\mathbf{\Phi x}_1} = {\mathbf{\Phi} \mathbf{x}_2}$ so that the OMP algorithm fails to recover $K$-sparse vector.
In fact, no reconstruction algorithm can always guarantee the perfect recovery of $K$-sparse vector under $\mu = \frac{1} {{2K - 1}}$.
\end{rem}

\end{document}